\documentclass[12pt]{article}

\usepackage[utf8]{inputenc}
\usepackage[T1]{fontenc}
\usepackage{amsmath, amssymb, amsfonts, mathtools, amsthm}
\usepackage{bbm}
\usepackage{booktabs}
\usepackage{multirow}
\usepackage{colortbl}
\usepackage{etoolbox}
\usepackage{siunitx}
\usepackage{color,soul}
\usepackage{xcolor}
\usepackage{geometry}
\usepackage{csquotes,nicefrac}
\usepackage[numbers,sort&compress]{natbib}
\usepackage{xurl} 
\usepackage[hidelinks]{hyperref}
\usepackage{cleveref}
\usepackage{graphicx}
\usepackage{subcaption}

\usepackage{microtype} 
\usepackage{titlesec}
\titlespacing*{\section}{0pt}{1.2ex plus 1ex minus .2ex}{0.5ex plus .2ex}
\titlespacing*{\subsection}{0pt}{1ex plus 1ex minus .2ex}{0.5ex plus .2ex}
\titlespacing*{\paragraph}{0pt}{1ex plus .5ex minus .2ex}{1em}

\DeclareMathOperator*{\argmax}{arg\,max}

\DeclarePairedDelimiterX{\inp}[2]{\langle}{\rangle}{#1, #2}

\newtheorem{theorem}{Theorem}

\title{\vspace{-15mm}Incentive Alignment in Online Experimentation}
\author{Ermis Soumalias, Richard Mudd, Abbas Zaidi
\\ \textit{Meta Platforms, Inc.}}
\date{\vspace{-8ex}}
\begin{document}

\maketitle
\vspace{-5 pt}

\begin{abstract}
Evaluating the causal effect of new features is a central goal for online platforms. While recent literature addresses limited testing traffic via centralized portfolio optimization, this perspective abstracts away a critical institutional reality: experimentation is operationally decentralized. The experimenters who develop new features also dictate which hypotheses to test, and they are typically rewarded based on empirical average treatment effects that are prone to upward bias. Left unchecked, this principal-agent conflict can severely erode platform value---a structural failure that conventional centralized levers, such as significance thresholds and traffic budgets, cannot resolve. By reframing experimentation as an incentive design problem, we demonstrate that two practical mechanisms, sample splitting and shrinkage, can effectively bridge this gap. Sample splitting aligns incentives perfectly at a bounded traffic cost, while shrinkage consumes no additional traffic and guarantees that interventions with negative expected effects are strictly unprofitable to field.
\end{abstract}

\section{Introduction}
Evaluating the causal effects of candidate interventions via A/B testing is a cornerstone of decision-making for online platforms. As these platforms face an ever-growing pool of ideas competing for scarce testing traffic, recent literature has predominantly modeled this bottleneck as a portfolio optimization problem \citep{schmit2019optimal,azevedo2020ab,sudijono2025optimizing}. Under this paradigm, a centralized decision-maker maximizes global returns by strategically trading off a few well-powered tests against a high volume of noisier estimates.

However, the theoretical appeal of this centralized approach obscures a practical institutional reality: in practice, resource allocation is highly decentralized. Individual experimenters---who possess the most informed local priors about their own features---typically dictate how testing budgets are spent \citep{tang2010overlapping,bakshy2014designing,kaufman2017democratizing}. Crucially, this decentralization introduces misaligned incentives. While the platform seeks to maximize the ground-truth value of launched features, experimenters are typically evaluated on the empirical average treatment effects of their successful tests. Because significance-filtered estimates \citep{vanzwet2021significance} are inherently susceptible to the winner's curse, experimenters face a rational incentive to favor testing strategies with high statistical upside, which can diverge significantly from genuine platform value.

To formalize this misalignment, we construct a budget allocation game. We demonstrate that when experimenters are rewarded for empirical effects, the resulting principal-agent conflict cannot be resolved through standard statistical gating. Indeed, for any fixed significance threshold and budget limit, strategic experimenters can select an allocation yielding exactly zero platform value, even when a platform-optimal allocation would generate arbitrarily large returns. Fundamentally, statistical gatekeeping alone cannot rescue structural misalignment.

To address this vulnerability, we reframe experimentation from a purely inferential exercise into a mechanism design problem. Specifically, we demonstrate how two practical mechanisms can successfully align experimenter behavior with platform value:

\begin{itemize}
    \item \textbf{Sample Splitting:} While traditionally viewed merely as an \emph{ex post} statistical tool to recover unbiased estimates, we provide a novel \emph{ex ante} economic rationale for its deployment. We prove that reserving a strictly independent evaluation sample perfectly aligns experimenter incentives, and that the resulting cost to overall platform value is strictly bounded.
    \item \textbf{Shrinkage:} Alternatively, practitioners often utilize shrinkage
to dampen marginal launches inflated by the winner's curse. We extend this statistical concept directly into the experimenter's economic incentive structure. By charging a variance penalty equal to the expected reward of a zero-effect intervention, we guarantee that negative-expected-value interventions become strictly unprofitable to test, all while preserving the full testing budget.
\end{itemize}

Taken together, our findings demonstrate that treating experimentation solely as a centralized optimization problem misses the practical behavioral dynamics of decentralized teams.
However, by viewing causal inference through a mechanism-design lens, platforms can safely decentralize experimentation and fully leverage the rapid pace of AI-driven feature generation without compromising the true value of their products.

\section{Related Work} \label{sec:related_work}

A growing literature models online experimentation as a resource-constrained portfolio problem, balancing well-powered tests against a larger volume of noisier estimates \citep{schmit2019optimal,azevedo2020ab,sudijono2025optimizing,feit2019test}. Crucially, what unites these treatments is the assumption of a single decision maker whose objective perfectly aligns with the platform. We retain the allocation framework but introduce strategic dynamics: the budget is spent by a self-interested experimenter maximizing an empirical proxy reward. Consequently, our work connects most closely to recent studies at the intersection of statistics and mechanism design \citep{hossain2025strategic,shi2024sharp,shi2025instance}, which model testing as a principal-agent game. Rather than deploying complex menus of statistical contracts, we demonstrate how to align incentives using practical mechanisms already familiar to practitioners: sample splitting and shrinkage.

Our sample-splitting mechanism (\Cref{subsec:sp_mechanism}) bridges the statistical literature on replication \citep{cox1975note,wasserman2009high} and selective inference \citep{fithian2014optimal,tian2018selective} with mechanism design. It is well established that conditioning on significance inflates reported effects \citep{lee2018winners,berman2022false}, an inferential hurdle formally characterized by \citet{andrews2024inference}.
While we borrow the sample-splitting machinery directly from this literature, our core contribution is the economic rationale: we establish its \emph{ex ante} utility for incentive alignment.

Finally, shrinking significance-filtered estimates toward a prior serves as a practical remedy for selection bias \citep{vanzwet2021significance,gelman2014beyond,kessler2024overcoming,mudd2026breaking}. Our shrinkage mechanism (\Cref{subsec:shrinkage}) applies this underlying intuition directly to the experimenter's incentive structure. By correcting the payoff itself rather than just the point estimate, we preserve the upside of large true effects while rendering null interventions strictly unprofitable (\Cref{thm:shrinkage_ic}).

\section{Formal Model} \label{sec:formal_model}
An experimenter allocates a budget of $B$ experimental units across $N$ candidate interventions ($b_i \in \mathbb{N}_{\ge0}$, $\sum b_i \leq B$). Intervention $i\in[N]$ has a true expected treatment effect $\mu_i\in\mathbb R$ and a known unit-level variance $\sigma_i^2>0$. Testing intervention $i$ with $b_i>0$ units yields an empirical effect $\hat{\mu}_i(b_i)\sim\mathcal N\!\left(\mu_i,\sigma_i^2/b_i\right)$.

We assume the sampling variances $\sigma_i^2$ are known and that the experimenter's prior belief matches the truth ($\tilde\mu_i=\mu_i$). This best-case information setting neatly isolates structural incentive misalignment from estimation error. At significance level $\alpha$ (with critical value $c_\alpha := z_{1-\alpha/2}$), the test is deemed significantly positive when
$L_i(b_i,\alpha) = \mathbbm{1}\!\left\{\hat{\mu}_i(b_i) > c_\alpha\frac{\sigma_i}{\sqrt{b_i}}\right\}$ is equal to $1$.
The resulting launch probability is:
\begin{equation} p_i(b) = \Pr\!\left[L_i(b,\alpha)=1\right] = \overline{\Phi}\!\left(c_\alpha - \frac{\mu_i\sqrt{b}}{\sigma_i}\right) \label{eq:launch_prob} \end{equation}
Because the platform values launched interventions solely by true effects, launching based on $L_i(b_i,\alpha)$ yields a realized platform reward of $\sum_{i=1}^N \mu_i L_i(b_i,\alpha)$.

\paragraph{The Incumbent Mechanism.} Under common industry practice, each intervention is launched upon a significantly positive result, and the experimenter's reward is proportional to the empirical effects of their launched interventions:
\begin{equation}
     R_E^{\mathrm{inc}}(\mathbf{b},\alpha)=\sum_{i=1}^N\hat{\mu}_i(b_i)L_i(b_i,\alpha).
\end{equation}
By comparison, the platform's reward is $R_P^{\mathrm{inc}}(\mathbf{b},\alpha)=\sum_{i=1}^N\mu_i L_i(b_i,\alpha)$.

For $b_i>0$, the experimenter's expected reward from intervention $i$ is the mean of a truncated Gaussian: \begin{equation} r_{E,i}^{\mathrm{inc}}(\mu_i,b_i) = \mathbb{E}\left[\hat{\mu}_i(b_i)L_i(b_i,\alpha)\right] = \mu_i p_i(b_i) + \frac{\sigma_i}{\sqrt{b_i}}\phi\left(c_\alpha-\frac{\mu_i\sqrt{b_i}}{\sigma_i}\right), \label{eq:experimenter_expected_reward_incumbent} \end{equation} while the platform's is $r_{P,i}^{\mathrm{inc}}(\mu_i,b_i)=\mu_i p_i(b_i)$. We denote the platform- and experimenter-optimal allocations as $\mathbf{b}_P^{\mathrm{inc}} \in \argmax_{\mathbf{b}} \mathbb{E}[R_P^{\mathrm{inc}}(\mathbf{b},\alpha)]$ and $\mathbf{b}_E^{\mathrm{inc}} \in \argmax_{\mathbf{b}} \mathbb{E}[R_E^{\mathrm{inc}}(\mathbf{b},\alpha)]$, respectively.

\section{Theoretical Results}

\subsection{Incentive Misalignment Can Completely Destroy Platform Value} \label{subsec:incentive_misalignment}

We first prove that no fixed statistical-significance threshold can eliminate this principal-agent conflict. Consequently, the resulting loss in platform value can be arbitrarily large.

\begin{theorem}
\label{thm:zero_platform_value}
For every $\alpha\in(0,1)$, $B\geq1$, and $V>0$, there exists an instance such that
$\mathbb{E}[R_P^{\mathrm{inc}}(\mathbf{b}_E^{\mathrm{inc}},\alpha)] \leq 0$ while
$\mathbb{E}[R_P^{\mathrm{inc}}(\mathbf{b}_P^{\mathrm{inc}},\alpha)]>V$.
\end{theorem}
The formal proof of \Cref{thm:zero_platform_value} is provided in \Cref{app:proof_zero_platform_value}.

\subsection{The Sample-Splitting Mechanism} \label{subsec:sp_mechanism}
To resolve this structural flaw, the \emph{sample-splitting mechanism} rewards the experimenter based on an independent, unbiased estimate of the launched interventions' true effects.

\paragraph{Sample-Splitting Mechanism.}
Fix a holdout fraction $\gamma\in(0,1)$. Partition the allocated $b_i$ units into a \emph{selection sample} of size $(1-\gamma)b_i$ and an independent \emph{evaluation sample} of size $\gamma b_i$. The launch decision $L_i=L_i((1-\gamma)b_i,\alpha)$ is made using only the selection sample. If launched, the evaluation sample generates a fresh estimate $\hat{\mu}_i^E$. The experimenter receives $R_E^{\mathrm{split}}(\mathbf{b},\alpha)=\sum_i\hat{\mu}_i^E L_i$, and the platform receives $R_P^{\mathrm{split}}(\mathbf{b},\alpha)=\sum_i \mu_i L_i$.

\begin{theorem} \label{thm:incentive_alignment}
For any budget allocation $\mathbf{b}$ and holdout fraction $\gamma \in (0,1)$,
the expected platform and experimenter rewards under the sample-splitting mechanism
satisfy:  $\mathbb{E}[R_E^{\mathrm{split}}(\mathbf{b},\alpha)] = \mathbb{E}[R_P^{\mathrm{split}}(\mathbf{b},\alpha)]$.
\end{theorem}
The formal proof is deferred to \Cref{app:proof_incentive_alignment}. Consequently, the optimal budget allocations for the platform and experimenter coincide ($\mathbf{b}_P^{\mathrm{split}} = \mathbf{b}_E^{\mathrm{split}}$). This alignment comes at a strictly bounded cost to overall statistical power.

\begin{theorem} \label{thm:power_loss}
For any holdout fraction $\gamma\in(0,1)$ and threshold $c_\alpha \leq 2$, the expected platform value under the sample-splitting mechanism satisfies $\mathbb{E}[R_P^{\mathrm{split}}(\mathbf{b}_E^{\mathrm{split}},\alpha)] \geq (1-\gamma) \mathbb{E}[R_P^{\mathrm{inc}}(\mathbf{b}_P^{\mathrm{inc}},\alpha)]$.
\end{theorem}
The formal proof is provided in \Cref{app:proof_power_loss}.

\subsection{The Shrinkage Mechanism} \label{subsec:shrinkage}
While sample splitting achieves perfect alignment by sacrificing some statistical power, the \textit{shrinkage} mechanism offers an alternative tradeoff. It preserves the full testing budget while partially aligning incentives, guaranteeing that non-positive expected-value interventions remain strictly unprofitable.

\paragraph{Shrinkage Mechanism.} For each intervention $i$ receiving $b_i>0$ units, the experimenter is rewarded based on the conditional empirical effect minus a variance penalty equal to the expected reward of a zero-effect intervention, $C_i(b_i,\alpha) = \frac{\sigma_i}{\sqrt{b_i}}\phi(c_\alpha)$. The net reward is:
\[ R_E^{\mathrm{shrink}}(\mathbf{b},\alpha) = \sum_{i:b_i>0} \left[ \hat{\mu}_i(b_i)L_i(b_i,\alpha)-C_i(b_i,\alpha) \right]. \]

\begin{theorem}[Sign Alignment] \label{thm:shrinkage_ic} For every intervention $i$ and allocation $b_i>0$, the expected reward under the shrinkage mechanism satisfies \[ \mathbb{E}\!\left[ \hat{\mu}_i(b_i)L_i(b_i,\alpha)-C_i(b_i,\alpha) \right] \begin{cases} >0, & \mu_i>0,\\ =0, & \mu_i=0,\\ <0, & \mu_i<0. \end{cases} \] Thus, under risk neutrality, negative-effect interventions are strictly unprofitable to test, while zero-effect interventions offer no expected gain from favorable experimental noise. \end{theorem}
The formal proof is deferred to \Cref{app:proof_shrinkage_ic}.

\begin{theorem}[Platform-Value Guarantee]
\label{thm:shrinkage_welfare}
For any fixed $\alpha\in(0,1)$, the platform's expected reward evaluated at the experimenter-optimal allocation in the shrinkage mechanism is a constant-fraction approximation of the maximum platform value achievable under the incumbent mechanism.
In particular, for standard significance $\alpha=0.05$, numerical evaluation gives an approximation factor of approximately $0.1772$.
\end{theorem}
The formal proof can be found in \Cref{app:proof_shrinkage_welfare}.

\section{Empirical Evaluation} \label{sec:experiments}
\begin{figure}[htbp]
    \centering
    \begin{subfigure}[t]{0.48\textwidth}
        \centering
        \includegraphics[width=0.95\linewidth]{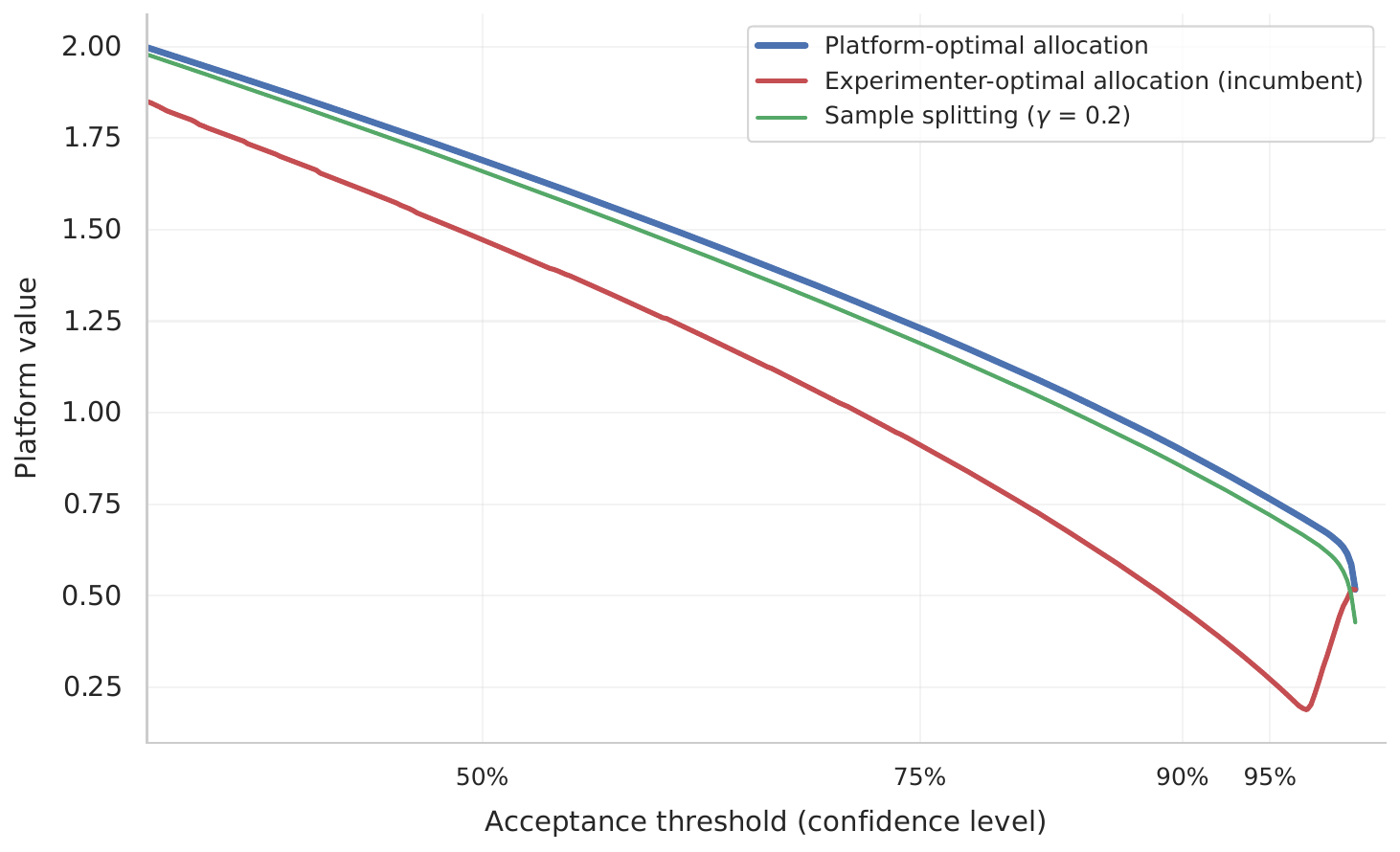}
        \caption{Instance 1: Perfect Priors}
        \label{fig:perfect_priors}
    \end{subfigure}
    \hfill
    \begin{subfigure}[t]{0.48\textwidth}
        \centering
        \includegraphics[width=0.95\linewidth]{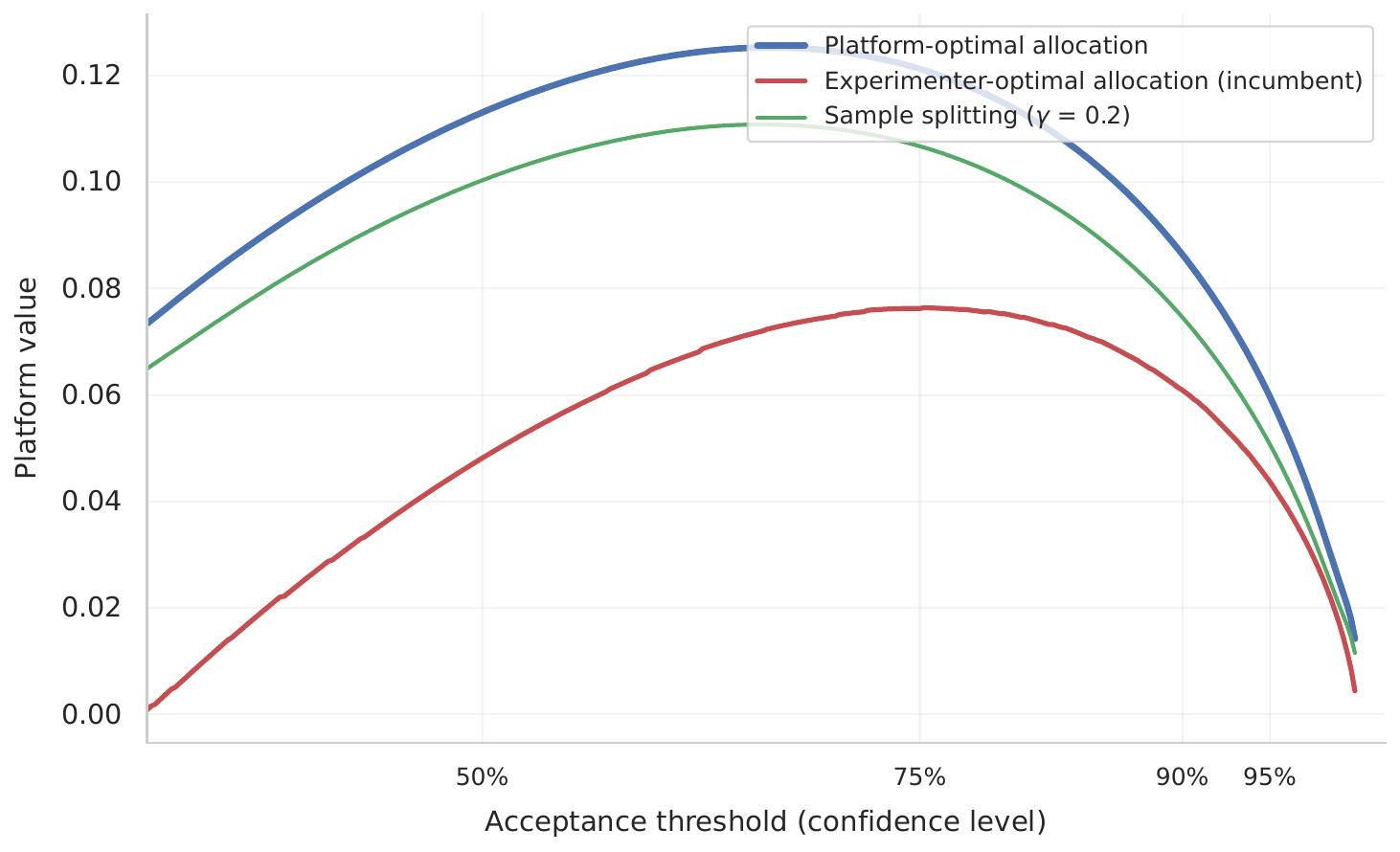}
        \caption{Instance 2: Noisy Priors}
        \label{fig:imperfect_priors}
    \end{subfigure}
    \caption{Expected platform value vs. acceptance threshold under platform-optimal allocation, experimenter-optimal allocation (incumbent mechanism), and sample splitting ($\gamma=0.2$).} \label{fig:results}
\end{figure}

We complement our theoretical findings with empirically calibrated simulations across two distinct instances ($N=1600$, $B=200{,}000$, minimum experiment size $500$ units). Distributing true expected treatment effects along a Gaussian mixture (a bulk of marginally harmful interventions alongside a thin tail of high-value winners), we evaluate sample splitting with a holdout fraction of $\gamma=0.2$. Instance 1 assumes perfect experimenter priors, whereas Instance 2 injects noise ($\tilde\mu_i=\mu_i+\varepsilon_i$, with $\varepsilon_i\sim\mathcal N(0,0.2^2)$) to rigorously test robustness against imperfect information.

In \Cref{fig:results}, we plot expected platform value against the acceptance threshold. Under perfect priors (Instance 1), a zero threshold is optimal for both the platform-optimal allocation and sample splitting: with aligned incentives, statistical filtering is entirely unnecessary. With noisy priors (Instance 2), thresholds become essential to screen out harmful interventions. Crucially, even when optimally tuned, the incumbent mechanism significantly underperforms the platform optimum, empirically reinforcing \Cref{thm:zero_platform_value}: statistical gatekeeping cannot rescue structural misalignment. 

Ultimately, these simulations highlight our core theoretical tradeoff. Sample splitting sacrifices a mere $1.0\%$ and $11.6\%$ of optimal platform value in Instances 1 and 2, respectively. In sharp contrast, strategic decentralization under the incumbent mechanism costs $7.4\%$ and a drastic $42.1\%$ of optimal platform value. Thus, sample splitting recovers the vast majority of value otherwise lost to decentralized decision-making.

\section{Conclusion}
Experimentation is operationally decentralized, driving a structural misalignment between platforms maximizing product value and experimenters optimizing significance-filtered empirical estimates. As we have demonstrated, this conflict is severe enough to drive program value to zero, and standard levers like significance thresholds fail to systematically resolve it. 

To address this failure, we provided novel economic interpretations for sample splitting and shrinkage, establishing their capacity to realign incentives \emph{ex ante}. Sample splitting achieves perfect alignment at a bounded cost to experimentation budget, whereas shrinkage fully preserves the budget while strictly deterring negative-expected-value interventions. Moving forward, platforms must navigate this fundamental tradeoff between budget efficiency and behavioral alignment to safely leverage decentralized, AI-driven feature generation.

\newpage
\small 
\setlength{\bibsep}{0pt}
\bibliographystyle{abbrvnat}
\bibliography{ref}
\normalsize

\newpage
\appendix
\section{Supplementary Material}

\subsection{Proof of \Cref{thm:zero_platform_value}} \label{app:proof_zero_platform_value}
\begin{proof}
Fix $\alpha,B,V$ and consider $B+1$ interventions: intervention $0$ has $(\mu_0,\sigma_0)=(M,1)$, while each $j\in \{1, \dots , B \}$ has $(\mu_j,\sigma_j)=(0,S)$.
The platform's expected reward for allocating budget $b >0$ to any $j\in\{1, \dots , B \}$ is zero, whereas for intervention 0, it is $M\overline{\Phi}(c_\alpha-M\sqrt b)>0$, which is monotonically non-decreasing in $b$. Thus, the platform-optimal allocation $\mathbf{b}_P^{\mathrm{inc}}$ assigns the entire budget $B$ to intervention 0, achieving an expected reward of $M\overline{\Phi}(c_\alpha-M\sqrt B)>V$ for sufficiently large $M$.

In contrast, for any zero-mean intervention $j\in\{1, \dots , B \}$, the experimenter's expected reward $r_{E,j}^{\mathrm{inc}}(0,b)=S\phi(c_\alpha)/\sqrt b$ is strictly decreasing in $b$. Consequently, the experimenter maximizes reward across these interventions by allocating exactly one unit to each. By choosing $S>\max_{1\leq k\leq B} \frac{r_{E,0}^{\mathrm{inc}}(M,k)}{k\phi(c_\alpha)}$, allocating any $k\geq1$ units to intervention 0 yields a total experimenter reward of at most $r_{E,0}^{\mathrm{inc}}(M,k)+(B-k)S\phi(c_\alpha) < BS\phi(c_\alpha)$. Thus, the experimenter-optimal allocation is $\mathbf{b}_E^{\mathrm{inc}}=(0,1,\ldots,1)$, which yields $\mathbb E[R_P^{\mathrm{inc}}(\mathbf{b}_E^{\mathrm{inc}},\alpha)]=0$.
\end{proof}

\subsection{Proof of \Cref{thm:incentive_alignment}} \label{app:proof_incentive_alignment}
\begin{proof}
Fix any feasible budget allocation $\mathbf{b}$, and let $L_i \in \{0,1\}$ indicate whether intervention $i$ is launched. By construction, $L_i$ depends only on the selection sample, while $\hat{\mu}_i^E$ is computed from the independent evaluation sample.
Thus, $\hat{\mu}_i^E$ and $L_i$ are independent.
Moreover, the evaluation-sample estimate is unbiased on i.i.d. samples, so $\mathbb{E}[\hat{\mu}_i^E] = \mu_i$. The experimenter's expected reward is therefore
$\mathbb{E}[R_E^{\mathrm{split}}(\mathbf{b}, \alpha)] =
\sum_{i=1}^N \mathbb{E}[\hat{\mu}_i^E]\mathbb{E}[L_i(b_i (1 - \gamma),\alpha)] =
\sum_{i=1}^N \mu_i p_i(b_i (1 - \gamma)) =
\mathbb{E}[R_P^{\mathrm{split}}(\mathbf{b} , \alpha)]$.
Since their expected rewards coincide for every feasible allocation $\mathbf{b}$, their sets of optimal allocations must also coincide under risk neutrality.
\end{proof}

\subsection{Proof of \Cref{thm:power_loss}} \label{app:proof_power_loss}
\begin{proof}
For any intervention with $\mu_i>0$, its launch probability given $b>0$ units is $p_i(b)=\overline{\Phi}(c_\alpha-\mu_i\sqrt b/\sigma_i)$, which becomes $p_i((1-\gamma)b)$ under the  sample-splitting mechanism.
We first show that $p_i((1-\gamma)b)\geq(1-\gamma)p_i(b)$. Let $a_i=\mu_i/\sigma_i>0$ and $x=a_i\sqrt b$.
Then, $p_i''(b)=\frac{a_i^4}{4x^3}\phi(c_\alpha-x)(c_\alpha x-x^2-1)$.
Thus, $p_i''(b)\leq0$ if and only if $x^2-c_\alpha x+1\geq0$. For $c_\alpha\leq2$, we have
$x^2-c_\alpha x+1=(x-c_\alpha/2)^2+1-c_\alpha^2/4\geq0$,
meaning $p_i(b)$ is concave. Since its continuous extension satisfies $p_i(0)=\overline{\Phi}(c_\alpha)\geq0$, concavity implies
$p_i((1-\gamma)b)\geq(1-\gamma)p_i(b)+\gamma p_i(0) \geq (1-\gamma)p_i(b)$.

The full-budget optimum allocates zero budget to interventions with $\mu_i<0$, as they would yield negative expected value. Applying the concavity inequality to all positively valued interventions in $\mathbf{b}_P^{\mathrm{inc}}$ gives
$\mathbb E[R_P^{\mathrm{split}}(\mathbf{b}_P^{\mathrm{inc}},\alpha)]\geq(1-\gamma)\mathbb E[R_P^{\mathrm{inc}}(\mathbf{b}_P^{\mathrm{inc}},\alpha)]$.
Finally, by \Cref{thm:incentive_alignment}, $\mathbf{b}_E^{\mathrm{split}}$ maximizes platform value under sample splitting, yielding
\begin{align}
\mathbb E[R_P^{\mathrm{split}}(\mathbf{b}_E^{\mathrm{split}},\alpha)] & =
\mathbb E[R_P^{\mathrm{split}}(\mathbf{b}_P^{\mathrm{split}},\alpha)] \nonumber \\
& \geq\mathbb E[R_P^{\mathrm{split}}(\mathbf{b}_P^{\mathrm{inc}},\alpha)] \nonumber \\
& \geq(1-\gamma)\mathbb E[R_P^{\mathrm{inc}}(\mathbf{b}_P^{\mathrm{inc}},\alpha)].
\end{align}
\end{proof}

\subsection{Proof of \Cref{thm:shrinkage_ic}} \label{app:proof_shrinkage_ic}
\begin{proof}
Fix an intervention $i$ and budget $b_i>0$.
Recall that its expected reward for the experimenter under the incumbent mechanism is
$r_{E,i}^{\mathrm{inc}}(\mu_i,b_i)
=
\mu_i p_i(b_i)
+
\frac{\sigma_i}{\sqrt{b_i}}
\phi\left(
c_\alpha-\frac{\mu_i\sqrt{b_i}}{\sigma_i}
\right),$
where $c_\alpha=z_{1-\alpha/2}$. The shrinkage charge is exactly the expected reward of a zero-effect intervention with the same variance as $i$:
$C_i(b_i,\alpha)
=
r_{E,i}^{\mathrm{inc}}(0,b_i)
=
\frac{\sigma_i}{\sqrt{b_i}}\phi(c_\alpha).$
Hence, the experimenter's expected reward for allocating $b_i$ experimental units to intervention $i$ under the shrinkage mechanism is:
$r_{E,i}^{\mathrm{shrink}}(\mu_i,b_i)  =  r_{E,i}^{\mathrm{inc}}(\mu_i,b_i)-r_{E,i}^{\mathrm{inc}}(0,b_i)$.

It immediately follows that $r_{E,i}^{\mathrm{shrink}}(0,b_i) = 0 $, meaning the experimenter has no incentive to launch experiments where the true effect is zero.

Finally, differentiating $r_{E,i}^{\mathrm{shrink}}(\mu,b_i)$ with respect to $\mu$:
\[
\frac{\partial r_{E,i}^{\mathrm{shrink}}(\mu,b_i)}{\partial \mu}
=
\frac{\partial r_{E,i}^{\mathrm{inc}}(\mu,b_i)}{\partial \mu} + 0
=
p_i(b_i)
+
c_\alpha
\phi\left(
c_\alpha-\frac{\mu\sqrt{b_i}}{\sigma_i}
\right)
>0,
\]
since $p_i(b_i)>0$, $c_\alpha>0$, and $\phi(\cdot)>0$. Thus,
\[
r_{E,i}^{\mathrm{shrink}}(\mu_i,b_i)
\begin{cases}
>0, & \mu_i>0,\\
=0, & \mu_i=0,\\
<0, & \mu_i<0,
\end{cases}
\]
which proves the claim.
\end{proof}

\subsection{Proof of \Cref{thm:shrinkage_welfare}} \label{app:proof_shrinkage_welfare}
\begin{proof}
Fix an intervention $i$ with $\mu_i>0$ and an allocation $b_i>0$, and define its signal-to-noise ratio as
$\lambda_i = \frac{\mu_i\sqrt{b_i}}{\sigma_i} >0.$
Its launch probability is
$p_i(b_i)
=
\overline{\Phi}(c_\alpha-\lambda_i).$
Hence, its expected contribution to the platform's reward is
$
r_{P,i}^{\mathrm{shrink}}(\mu_i,b_i) =
r_{P,i}^{\mathrm{inc}}(\mu_i,b_i) =
\mu_ip_i(b_i) = \mu_i\overline{\Phi}(c_\alpha-\lambda_i).$

From \Cref{eq:experimenter_expected_reward_incumbent} under the incumbent mechanism, the experimenter's expected reward from this intervention is
$r_{E,i}^{\mathrm{inc}}
(\mu_i,b_i)
=
\mu_ip_i(b_i)
+
\frac{\sigma_i}{\sqrt{b_i}}
\phi(c_\alpha-\lambda_i).$
The shrinkage mechanism subtracts
$C_i(b_i,\alpha)
=
\frac{\sigma_i}{\sqrt{b_i}}\phi(c_\alpha).$
Thus, the experimenter's expected reward under the shrinkage mechanism is
\begin{equation}
r_{E,i}^{\mathrm{shrink}}(\mu_i,b_i)
=
\mu_ip_i(b_i)
+
\frac{\sigma_i}{\sqrt{b_i}}
\left[
\phi(c_\alpha-\lambda_i)-\phi(c_\alpha)
\right].
\end{equation}

Since
$\mu_i
=
\lambda_i\frac{\sigma_i}{\sqrt{b_i}},$
we can rewrite the platform's expected reward as
$r_{P,i}^{\mathrm{shrink}}(\mu_i, b_i)
= \mu_ip_i(b_i)
=
\frac{\sigma_i}{\sqrt{b_i}}
\lambda_i
\overline{\Phi}(c_\alpha-\lambda_i).$
Thus,
\begin{align}
\frac{r_{E,i}^{\mathrm{shrink}}(\mu_i,b_i) }{r_{P,i}^{\mathrm{shrink}}(\mu_i,b_i)}
&=
1+
\frac{
\frac{\sigma_i}{\sqrt{b_i}}
\left[
\phi(c_\alpha-\lambda_i)-\phi(c_\alpha)
\right]
}{
\frac{\sigma_i}{\sqrt{b_i}}
\lambda_i\overline{\Phi}(c_\alpha-\lambda_i)
}
\nonumber\\
&=
1+
\frac{
\phi(c_\alpha-\lambda_i)-\phi(c_\alpha)
}{
\lambda_i\overline{\Phi}(c_\alpha-\lambda_i)
}
\nonumber\\
&=:
q_\alpha(\lambda_i).
\end{align}

We next demonstrate that $q_\alpha(\lambda) \in \mathbb{R}_{> 0}$ for all $\lambda>0$.
Let $m_\alpha = \inf_{\lambda > 0} q_{\alpha} (\lambda)$ and $M_{\alpha} = \sup_{\lambda > 0} q_{\alpha} (\lambda) $.
First, by \Cref{thm:shrinkage_ic},
$r_{E,i}^{\mathrm{shrink}}(\mu_i,b_i) > 0  $ whenever $\mu_i>0$.
Since $r_{P,i}^{\mathrm{shrink}}(\mu_i,b_i) = \mu_ip_i(b_i) > 0$ as well, $q_\alpha(\lambda)>0 \;\forall \lambda>0.$

To see that $M_\alpha<\infty$, first note that
$q_\alpha(\lambda)$ has the finite limit
\begin{align}
\lim_{\lambda\to0^+}q_\alpha(\lambda) & =
1 + \lim_{\lambda\to0^+}
\frac{
\phi(c_\alpha-\lambda)-\phi(c_\alpha)
}{
\lambda\overline{\Phi}(c_\alpha-\lambda)
}
\nonumber\\
&=
1 + \lim_{\lambda\to0^+}
\frac{
(c_\alpha-\lambda)\phi(c_\alpha-\lambda)
}{
\overline{\Phi}(c_\alpha-\lambda)
+
\lambda\phi(c_\alpha-\lambda)
}
\nonumber\\
&=
1 + \frac{
c_\alpha\phi(c_\alpha)
}{
\overline{\Phi}(c_\alpha)
},
\end{align}
where the second equality follows from L'H\^opital's rule using $\phi'(x)=-x\phi(x)$.

Hence, by continuity, $q_\alpha(\cdot)$ is bounded on $(0,\delta]$ for some $\delta>0$.
For $\lambda\geq\delta$,
$\lambda\overline{\Phi}(c_\alpha-\lambda)
\geq
\delta\overline{\Phi}(c_\alpha-\delta)>0,$
while $\phi$ is bounded. Thus $q_\alpha(\lambda)$ is also bounded on
$[\delta,\infty)$, and therefore $M_\alpha<\infty$.
Thus,
\[
0
<
m_\alpha
\leq
M_\alpha
<
\infty.
\]

It follows that, for every positive-effect intervention $i$,
\[
m_\alpha r_{P,i}^{\mathrm{shrink}} (\mu_i,b_i)
\leq
r_{E,i}^{\mathrm{shrink}} (\mu_i,b_i)
\leq
M_\alpha r_{P,i}^{\mathrm{shrink}} (\mu_i,b_i).
\]

We now extend this intervention-level bound to complete budget allocations.
By \Cref{thm:shrinkage_ic}, interventions with negative and zero true effects yield strictly negative and zero expected reward for the experimenter under the shrinkage mechanism, respectively.
Thus, an experimenter-optimal allocation will never allocate budget to a negative-effect intervention.
Similarly, a platform-optimal allocation never benefits from allocating budget to any intervention with non-positive true effect.
Thus, for any relevant allocation $\mathbf{b}$:
\begin{equation}
\mathbb{E}[R_E^{\mathrm{shrink}}(\mathbf{b},\alpha)]
=
\sum_{i:\mu_i>0} r_{E,i}^{\mathrm{shrink}}(\mu_i,b_i)
\geq
m_\alpha
\sum_{i:\mu_i>0} r_{P,i}^{\mathrm{shrink}}(\mu_i,b_i)
=
m_\alpha
\mathbb{E}[R_P^{\mathrm{shrink}}(\mathbf{b},\alpha)],
\label{eq:shrink_reward_inequality}
\end{equation}
and similarly,
$\mathbb{E}[R_E^{\mathrm{shrink}}(\mathbf{b},\alpha)]
\leq
M_\alpha
\mathbb{E}[R_P^{\mathrm{shrink}}(\mathbf{b},\alpha)].$

Let
$\mathbf{b}_E^{\mathrm{shrink}}$ and $\mathbf{b}_P^{\mathrm{shrink}}$
denote experimenter- and platform-optimal allocations under the shrinkage mechanism, respectively. By definition,
$\mathbb{E}\left[
R_E^{\mathrm{shrink}}(\mathbf{b}_E^{\mathrm{shrink}},\alpha)
\right]
\geq
\mathbb{E}\left[
R_E^{\mathrm{shrink}}(\mathbf{b}_P^{\mathrm{shrink}},\alpha)
\right].$

Putting everything together,
\begin{align}
    M_\alpha \mathbb{E}\left[ R_P^{\mathrm{shrink}}(\mathbf{b}_E^{\mathrm{shrink}},\alpha) \right] &  \ge
    \mathbb{E}\left[ R_E^{\mathrm{shrink}}(\mathbf{b}_E^{\mathrm{shrink}},\alpha) \right] \nonumber \\
    & \ge \mathbb{E}\left[ R_E^{\mathrm{shrink}}(\mathbf{b}_P^{\mathrm{shrink}},\alpha) \right] \nonumber \\
    & \ge m_{\alpha}  \mathbb{E}\left[ R_P^{\mathrm{shrink}}(\mathbf{b}_P^{\mathrm{shrink}},\alpha) \right]. \nonumber
\end{align}

Since $M_\alpha>0$,
\begin{equation*}
\mathbb{E}\left[
R_P^{\mathrm{shrink}}(\mathbf{b}_E^{\mathrm{shrink}},\alpha)
\right]
\geq
\frac{m_\alpha}{M_\alpha}
\mathbb{E}\left[
R_P^{\mathrm{shrink}}(\mathbf{b}_P^{\mathrm{shrink}},\alpha)
\right].
\end{equation*}
Since shrinkage changes only the experimenter's payoff, not the launch rule, $R_P^{\mathrm{shrink}}(\mathbf{b},\alpha)=R_P^{\mathrm{inc}}(\mathbf{b},\alpha)$ for every allocation $\mathbf{b}$, so $\mathbf{b}_P^{\mathrm{shrink}}=\mathbf{b}_P^{\mathrm{inc}}$ and the right-hand side equals $\frac{m_\alpha}{M_\alpha}\mathbb{E}[R_P^{\mathrm{inc}}(\mathbf{b}_P^{\mathrm{inc}},\alpha)]$.

For $\alpha=0.05$, numerical evaluation yields $m_\alpha \approx 0.9891$ and $M_\alpha \approx 5.5820$, and therefore $\frac{m_\alpha}{M_\alpha} \approx 0.1772$, completing the proof.
\end{proof}

\end{document}